\documentclass{article}
\usepackage[a4paper, total={6in, 8in}]{geometry}
\usepackage{graphicx,amssymb,amsmath}
\usepackage{amsthm}
\usepackage{lscape}

\newcommand{\Z}{{\mathbb Z}}
\newcommand{\N}{{\mathbb N}}
\newcommand{\R}{{\mathbb R}}
\newcommand{\W}{{\mathcal W}}

\newcommand{\E}{{\mathcal E}}
\renewcommand{\H}{{\mathcal H}}

\newtheorem{definition}{Definition}

\newtheorem{theorem}{Theorem}
\newtheorem{cor}{Corollary}
\newtheorem{obs}{Observation}
\newtheorem{lemma}{Lemma}

\title{Helly Numbers of Polyominoes\thanks{This paper was published in Graphs and Combinatorics, September 2013, Volume 29, Issue 5, pp 1221-1234~\cite{cikl-hnp-11j}}}

\author{Jean Cardinal\thanks{Computer Science Department, 
Universit\'e Libre de Bruxelles (ULB), Belgium, 
{\tt \{jcardin, mkormanc, stefan.langerman\}@ulb.ac.be}}
\and 
Hiro Ito\thanks{
School of Informatics, Kyoto University, Japan, 
{\tt itohiro@kuis.kyoto-u.ac.jp}}\and 
Matias Korman\footnotemark[2] \and 
Stefan Langerman\footnotemark[2] $^,$\thanks{Directeur de Recherches du
F.R.S.-FNRS}}

\begin{document}
\thispagestyle{empty}
\maketitle

\begin{abstract}
We define the Helly number of a polyomino $P$ as the smallest number $h$ such that the $h$-Helly property holds for the family of symmetric and translated copies of $P$ on the integer grid. We prove the following: (i) the only polyominoes with Helly number 2 are the rectangles, (ii) there does not exist any polyomino with Helly number 3, (iii) there exist polyominoes of Helly number $k$ for any $k\neq 1,3$. 
\end{abstract}

\section{Introduction}

Helly's theorem on convex sets is a cornerstone of discrete geometry, with countless corollaries and extensions in both geometry and combinatorics. For instance, Helly-type properties of convex lattice subsets and hypergraphs have been studied since the 70's~\cite{d-ccl-73}. On the other hand, the theory of polyominoes, connected subsets of the square lattice $\mathbb{Z}^2$, has been developed since the 50's with the seminal works of Solomon Golomb~\cite{Golomb} and the famous recreational mathematician Martin Gardner.

In this paper, we propose a natural definition of the Helly number of a polyomino $P$ by considering families of symmetric and translated copies of $P$. We show that the only polyominoes with Helly number 2 are rectangles. We prove the surprising fact that there does not exist any polyomino with Helly number 3. Finally, we exhibit polyominoes of Helly number $k$ for any $k \geq 4$. Since there cannot be polyominoes of Helly number $1$, this completely characterizes the values of $k$ for which there exist polyominoes with Helly number $k$.

\subsection*{Definitions}

We define a planar graph $G=(\Z^2, E)$ that represents the adjacency relation between grid points. Each vertex $(i,j)$ is connected to its four neighbors $(i, j-1)$, $(i-1, j)$, $(i+1,j)$, and $(i, j+1)$. A subset of $\Z^2$ is \emph{connected} if its induced subgraph in $G$ is connected.

\begin{definition}
A \emph{polyomino} is a connected finite subset of $\Z^2$.
\end{definition}
\begin{figure}
\center
\includegraphics[width=0.7 \textwidth]{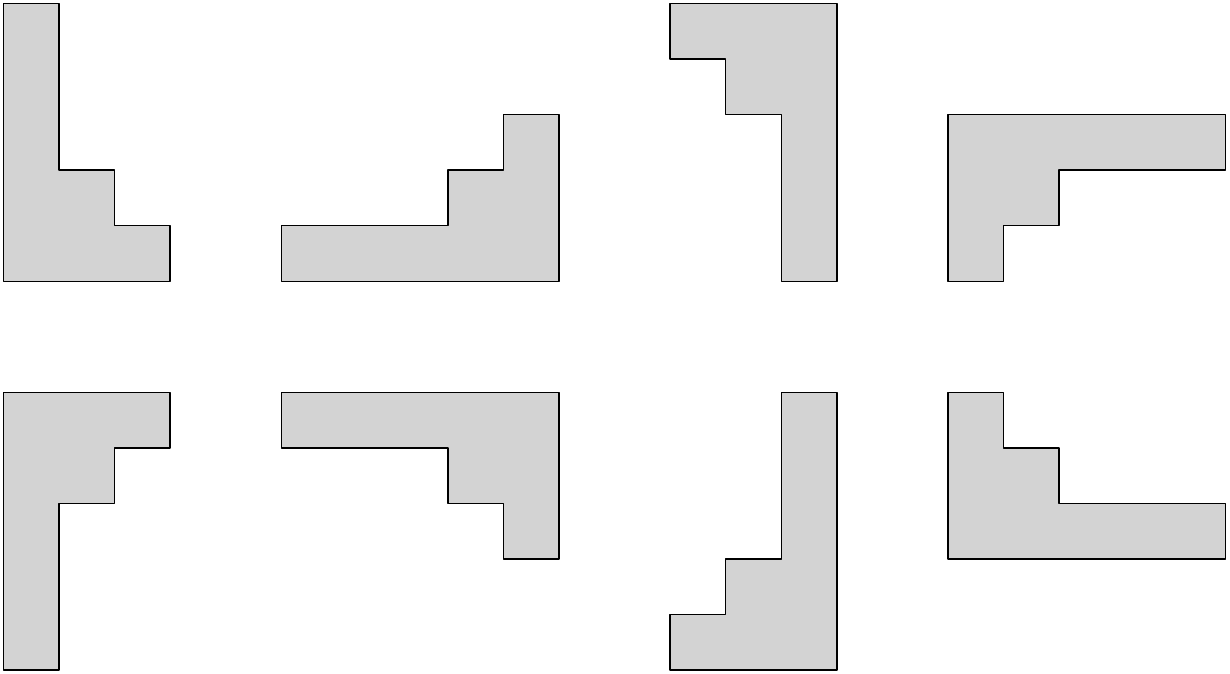}
\caption{Eight possible symmetries of a polyomino.}
\label{fig_sym}
\end{figure}

We often identify the point $(i,j) \in \Z^2$ with the unit square $[i,i+1]\times[j,j+1]\subset \R^2$. With this transformation a polyomino becomes an orthogonal polygon whose edges are on the unit grid. A {\em copy} of a polyomino $P$ is the image of $P$ by the composition of an integer translation with one of the eight symmetries of the square (that is, a mirror image and/or a $90$, $180$, or $270$-degree rotation of $P$). Figure~\ref{fig_sym} shows an example of a polyomino and its eight symmetries.
The cardinality of a polyomino will be denoted by $|P|$ (and will be referred as the {\em size} of $P$).

\begin{definition}
For any $k\in \mathbb{N}$, we say that polyomino $P$ satisfies the $k$-Helly\cite{gr-hdcg-04} property if, for any finite family $\mathcal{A}$ of copies of $P$ in which $A_1 \cap \ldots \cap A_k \neq \emptyset$  for any $A_1,\ldots, A_k\in \mathcal{A}$, we have $\bigcap_{A \in \mathcal{A}} A\neq \emptyset$. The {\em Helly number} $\H (P)$ of a polyomino $P$  is the smallest $k\in\N$ such that $P$ is $k$-Helly. 
\end{definition}


By definition, any polyomino $P$ that satisfies the  $k$-Helly property will also satisfy $k'$-Helly (for any $k'\geq k$). 

\subsection*{Previous work}

A {\em convex lattice set} in $\Z^d$ is the intersection of a convex set in $\R^d$ with the integer grid $\Z^d$. In 1973, Doignon proved that any family of convex lattice sets in $\Z^d$ is $2^d$-Helly~\cite{d-ccl-73}. A matching lower bound is obtained by considering all subsets of size $2^d-1$ of $\{0,1\}^d$. In our context, this implies that any convex polyomino (i.e. a polyomino that is the intersection a convex set in $\R^2$ with $\Z^2$) is $4$-Helly. Note that this is different from the term {\em convex polyomino}, which usually refers to polyominoes that are simultaneously row and column convex. 

Fractional Helly numbers of convex lattice subsets are studied by B\'ar\'any and Matousek~\cite{Barany2003}. Recently, Golumbic, Lipshteyn, and Stern showed that 1-bend paths on a grid have Helly number 4~\cite{GolumbicLS09}\footnote{In fact the paper claims that 1-bend paths have Helly number 3. However, through personal communication, we heard that there is an error in the paper. This error will be corrected in an upcoming publication by the authors.}. We note the environment considered is slightly different, since they considered that two paths have nonempty intersection whenever they share an edge. 

Recently, Cardinal {\em et. al}~\cite{ccisklt-cagnvttt-11b} consider a variation of the well-known Tic-Tac-Toe game in which the first player occupies a cell of the grid, and the second player locates a copy of a given polyomino $P$. The objective of the first player is to construct a copy of $P$ while the second player must prevent him from doing so. Among other results, the authors of \cite{ccisklt-cagnvttt-11b} give a general winning strategy for the first player, provided that $P$ is a polyomino that satisfies the $2$-Helly property. Unfortunately, in this paper we show that their strategy cannot be used with many polyominoes, since only rectangles are $2$-Helly.
 
\section{Helly Number up to 4}
In this Section we study polyominoes of small Helly number. Since we are considering finite polyominoes, it is easy to see that no polyomino can have Helly number $1$. Thus, we first look for polyominoes with Helly number two.

\begin{definition}
A \emph{rectangle} in $\Z^2$ is the cartesian product of two intervals in $\Z$. 
\end{definition}

It is easy to see that rectangles have Helly number 2. We show that the converse also holds.

\begin{theorem}\label{theo:H2}
A polyomino has Helly number 2 if and only if it is a rectangle.
\end{theorem}

In the following we give a slightly stronger result; we will show that the only polyominoes that satisfy the $3$-Helly property are rectangles.

\begin{definition}
A polyomino $P$ has the \emph{small empty quadrant} structure if for some
copy $P'$ of $P$, there exist values 
$x_1,y_1\in\Z$ such that the intersection of $P'$ with the 
$2\times 2$ rectangle $[x_1,x_1+1]\times[y_1-1,y_1]$ has cardinality
$\geq 3$, and $P'$ contains no point in $\{(x,y):x \geq x_1, y > y_1\}$ (see Figure \ref{fig_littlebig} $(a)$).
\end{definition}

\begin{definition}
A polyomino $P$ has the \emph{big empty quadrant} structure if for some
copy $P'$ of $P$, there exist values
$x_1,y_1,x_2,y_2\in\Z$, $y_1<y_2$, $x_1<x_2$ such that $\{(x_1,y_2),(x_1,y_1),(x_2,y_1)\}\subset P'$ and $P'$ contains no point in the upper right quadrant $\{(x,y):x>x_1, y>y_1\}$ (see Figure \ref{fig_littlebig} $(b)$).
\end{definition}

\begin{figure}
\center
\includegraphics[width=0.7\textwidth]{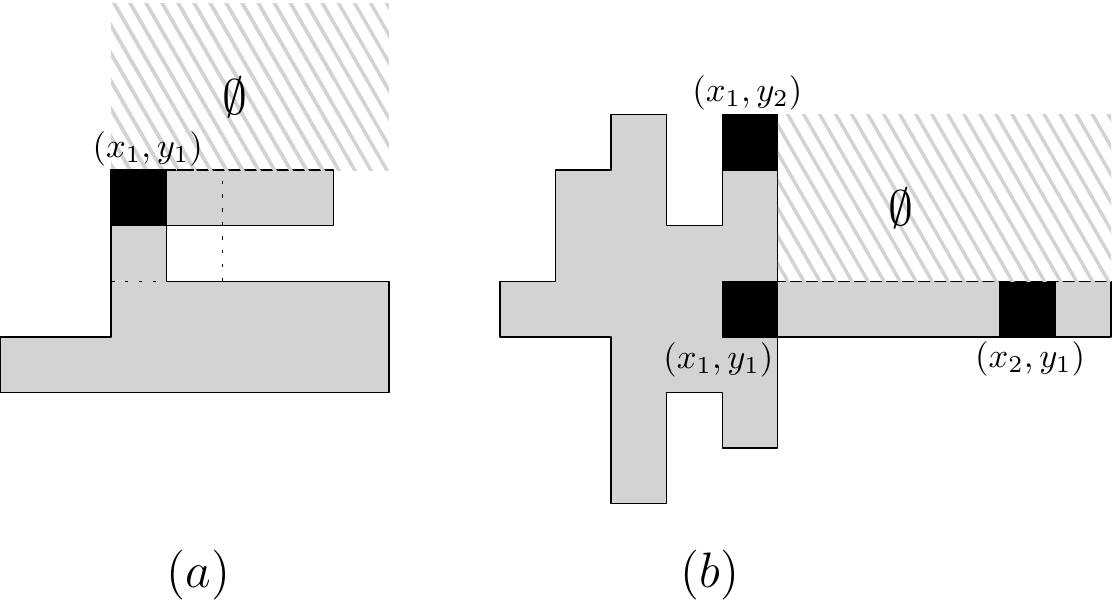}
\caption{Illustration of Lemma \ref{lem:bigsmallqu}. In order for a polyomino $P$ (of height at least 2) to not have the small empty quadrant structure (case $(a)$), $P$ cannot have two consecutive points on its upper boundary. If this  occurs, we can find a big empty quadrant (case $(b)$). The coordinates $x_1,x_2,y_1,y_2$ that generate the big or small empty quadrant are shown in black.}
\label{fig_littlebig}
\end{figure}


Given a rectangle $[x_0,x_1]\times [y_0,y_1]$, its {\em height} is $y_1-y_0+1$. Analogously, its {\em width} is  $x_1-x_0+1$. The height and width of a polyomino $P$ are equal to the height and width of the minimal bounding rectangle of $P$ (i.e.,the smallest rectangle in $\Z^2$ that contains $P$). 

\begin{lemma}\label{lem:bigsmallqu}
Every polyomino $P$ of height and width at least 2 either has the small empty quadrant or the big empty quadrant structure.
\end{lemma}
\begin{proof}
Observe that if $P$ has either height or width exactly $1$ it must be a rectangle. Hence,  this Lemma shows that any polyomino (other than some rectangles), has one of the two structures. An sketch of the proof of the claim is as follows: let $(i,j)$ be a point on the upper boundary of $P$ with at least two neighbors in $P$ (say points $(i-1,j)$ and $(i,j-1)$). Since $(i,j)$ is a boundary point, there will be a quadrant adjacent to it that is empty. In particular, $P$ will have a small empty quadrant structure. Thus, in order for $P$ to not have this structure, there cannot be a point on the upper, lower, right or left boundary of $P$ with two or more neighbors. However, in this situation we will show that $P$ must contain the big empty quadrant structure. 

 Let $(x_0,y_0)$ be the point of $P$ highest $x$-coordinate along the upper boundary of its bounding box. We will first show that if $(x_0,y_0-1)\not\in P$, then there exists $i\in\N$ such that $(x_0-i+1, y_0),(x_0-i,y_0),(x_0-i,y_0-1)\in P$. Proof of this claim is as follows: by definition of $(x_0,y_0)$, we have that $(x_0+1, y_0)\not\in P$, and $(x_0, y_0+1)\not\in P$. If we suppose that $(x_0,y_0-1)\not\in P$, then, in order for $P$ to be connected, we must have $(x_0-1,y_0)\in P$. By applying the same argument iteratively on this new point, we must have that eventually there exists an $i$ such that both $(x_0-i-1,y_0)\in P$ and $(x_0-i-1,y_0-1)\in P$, otherwise $P$ is a rectangle of height 1.

Therefore, if $(x_0,y_0-1)\not\in P$, $P$ has the small empty quadrant structure. Now assume otherwise and let $j$ be the smallest integer such that $(x_0,y_0-j)\in P$ and $(x_0,y_0-j-1)\not\in P$. If the quadrant $\{(x,y):x>x_0, y\geq y_0-j\}$ contains no point of $P$, then, by the same  argument as in the above claim, there must be a point of $P$ immediately left of the column $x_0$ between $y_0$ and $y_0-j$. In other words, there must be an integer $j'\in[0,j-1]$ such that $| P\cap ([x_0-1,x_0]\times[y_0-j'-1,y_0-j']) |\geq 3$,
and again $P$ has the small empty quadrant structure.

Finally, if the quadrant $\{(x,y): x>x_0, y\geq y_0-j\}$ is not empty, let $(x',y')$ be the highest point in that quadrant (pick one arbitrarily if many exist). In that case, the three points $(x_0,y_0), (x_0,y'), (x',y')$ form a big empty quadrant structure. 
\end{proof}

\begin{figure}
	\begin{center}
	\includegraphics[width=0.5\textwidth]{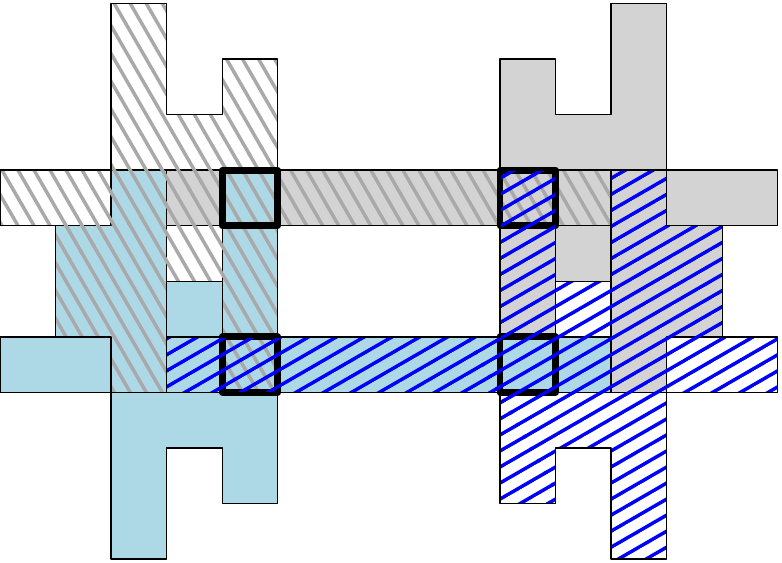}
	\caption{Proof of Lemma \ref{lem:beqex}. By flipping and rotating $P$, we obtain four copies that every three intersect, but there is no point common to the four copies. The highlighted unit squares correspond to points $(x_1,y_1), (x_1,y_2),(x_2,y_1)$ and $(x_2,y_2)$.}
	\label{fig_big}
	\end{center}
\end{figure}
\begin{lemma}\label{lem:beqex}
If a polyomino $P$ has the big empty quadrant structure, then $\H (P)\geq 4$.
\end{lemma}
\begin{proof}
We construct an arrangement of four copies of $P$ such that every subset of three copies have a common point, but there is no point common to all four copies. 

Consider the three points $(x_1,y_2)$, $(x_1,y_1)$, and $(x_2,y_1)$ given by the big empty quadrant structure in $P$. We construct the copies by flipping $P$ around the $x$ and/or $y$ axis so that those three points map to all possible triples of points in the set $\{ (x_1, y_1), (x_1, y_2), (x_2, y_1), (x_2, y_2) \}$. Since $(x_2, y_2)\not\in P$, each of the four points is missing from exactly one copy $P_i$, but belongs to the other three (see Figure \ref{fig_big}). 

Now we observe that the empty quadrants of the four copies cover $\Z^2$. Hence for any $(x,y)\in\Z^2$, there exists at least one $i\in\{1,2,3,4\}$ such that $(x,y)\not\in P_i$. Therefore, the four copies have no common intersection point.
\end{proof}

We now consider polyominoes that have the small empty quadrant structure. We will use the following observation. 

\begin{obs}
\label{obs:El}
For any polyomino $P$ that is not a rectangle, there exists a $2\times 2$ rectangle $R$ such that $|P\cap R|=3$.
\end{obs}

\begin{figure}[tb]
	\begin{center}
	\includegraphics[width=0.6\textwidth]{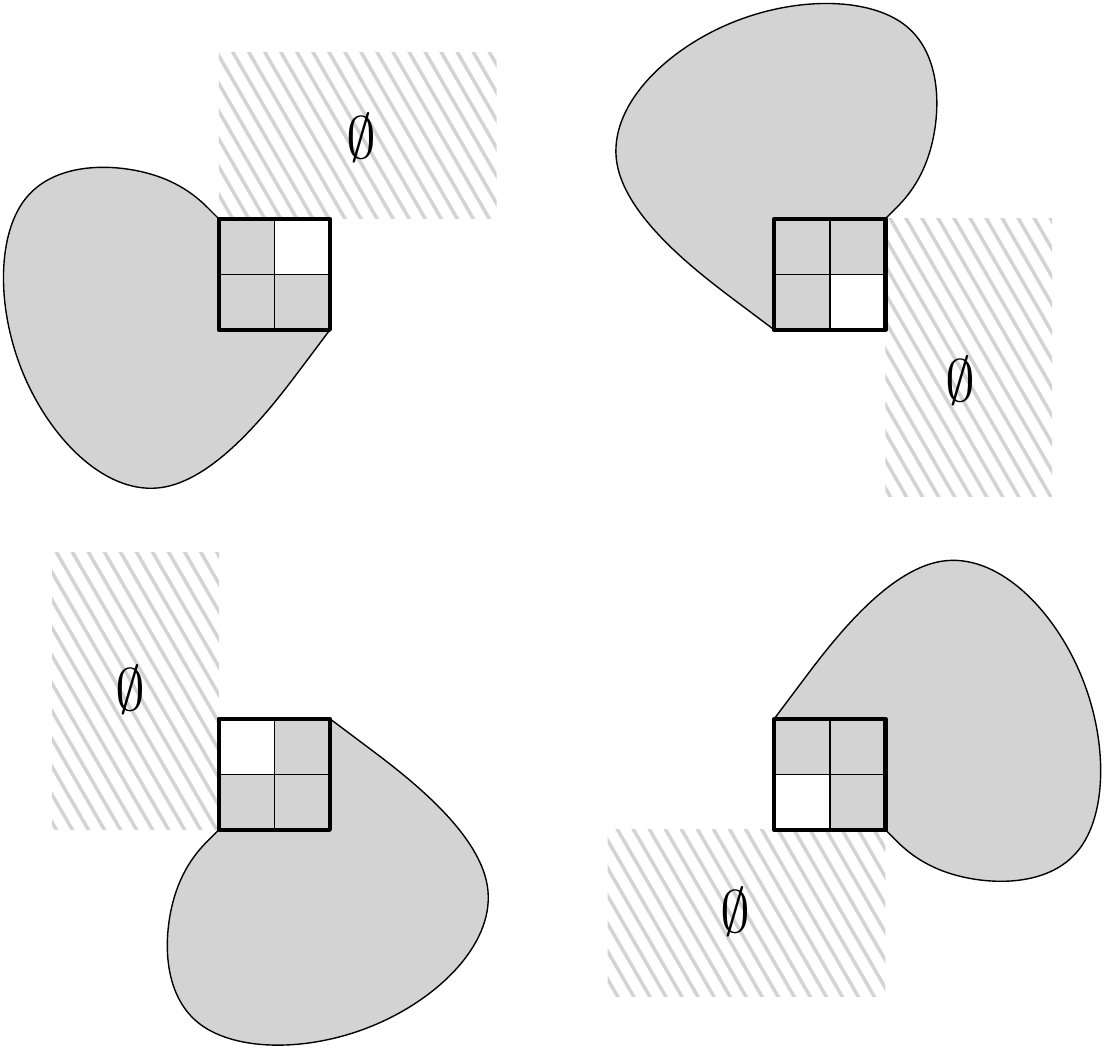}
	\caption{llustration of Lemma \ref{lem_small} for the case in which $I$ has cardinality 3 (denoted by a thick square). By translating the copies so that the respective squares overlap, we obtain a set in which any three copies have nonempty intersection. Since there is no common intersection point, $P$ cannot satisfy the $3$-Helly property.}
	\label{fig_small}	
	\end{center}
\end{figure}
\begin{lemma}\label{lem_small}
\label{lem:seqex}
If a polyomino $P$ has the small empty quadrant structure and is not a rectangle, then $\H (P)\geq 4$.
\end{lemma}
\begin{proof}
We construct an arrangement of at most 8 copies of $P$ such that every subset of three copies have a common point, but there is no point common to all copies. Let $(x_1, y_1)$ be the point given by the small empty quadrant structure, and $P'$ the corresponding copy of $P$.

We first consider the case in which the intersection $I$ of $P'$ with the $2\times 2$ rectangle $[x_1,x_1+1]\times[y_1-1,y_1]$ has cardinality exactly 3. In that case, we can use a similar construction as in Lemma~\ref{lem:beqex}, with four copies of $P$; we define the copies $P_i$ for $i=1,2,3,4$ as the four rotations of $P$ that map the bounding box of $I$ to the same $2\times 2$ rectangle. Those four points are the respective intersection points of all four possible triples. Similar to the previous case, the four empty quadrants cover all the other points of $\Z^2$, hence there cannot be a common intersection point (see Figure \ref{fig_small}).

It remains to consider the case in which the intersection $I$ has size 4. In this situation we use the same construction, but complete it with four more copies. From Observation~\ref{obs:El} and the fact that  $P$ is not a rectangle, we know that there exists a $2\times 2$ rectangle $R$ such that $|P'\cap R|=3$. We add four additional copies $P_i$, with $i=5,6,7,8$, that are the four rotations of a translated copy of $P'$ mapping $R$ to the bounding box of $L$. Each of the four points of this rectangle belongs to copies $P_1,P_2,P_3,P_4$ (since $|L|=4$), and to exactly three of the four copies $P_5, P_6, P_7, P_8$ (since $|P'\cap R|=3$). Hence every triple of copies intersects. However, from the previous construction, there still exists no point common to all 8 copies. We note that the above construction cannot be used if $P$ is a rectangle, since Observation~\ref{obs:El} does not hold in that case. 
\end{proof}

\begin{cor}\label{cor:H3}
There is no polyomino of Helly number 3.
\end{cor}

Combining this result with the upper bound of \cite{d-ccl-73}, we can compute the Helly number of any convex polyomino:

\begin{cor}\label{cor_convexpoly}
Let $P$ be a polyomino that is the intersection a convex set in $\R^2$ with $\Z^2$. If $P$ is a rectangle then $\H (P)= 2$. Otherwise $\H (P)= 4$. 
\end{cor}

\section{Hypergraph Generalization}\label{sec_upper}
In this section we study some interesting properties of polyominoes of Helly number $k$. Since these results hold for subsets of a discrete set of points, we state them in a more general fashion. Instead of copies of a given polyomino we can consider the same definitions for families of subsets of $\Z^2$. Using this idea, one can extend the Helly property to hypergraphs.

\begin{definition}
A hypergraph $G=(V,\E)$ is $k$-Helly if for any $\W \subseteq \E$ such that  $e_1 \cap \ldots \cap e_k \neq \emptyset$ for all $e_1,\ldots, e_k \in \W$, we have $\cap
_{e\in \W} e \neq \emptyset$. The {\em Helly number} $\H(G)$ of a hypergraph $G$ is the smallest value $k$ such that $G$ is $k$-Helly. 
\end{definition}

Observe that the above definition is a generalization of the previous definition for the polyomino case. Indeed, the polyomino formulation is the particular case in which $V=\Z^2$ and $\E$ contains all subsets of points contained in copies of a fixed polyomino $P$. Helly numbers of hypergraphs have been deeply studied; see for example the book of Graham, Gr\"{o}tschel, and Lov\'{a}sz (\cite{ggl-hoC-95}, Chapters 2 and 4), or the book of Berge (\cite{Berge}, Chapter 1), where relationship between conformal and 2-Helly hypergraphs is studied. There has also been a strong interest in computational aspects of this problem (like for example recognition); see the survey of Dourado, Protti, and Szwarcfiter \cite{hypersurvey}. Also see the paper of Barbosa {\em et al.}~\cite{colorHelly}, in which the chromatic variant of the Helly property is studied.

Let $G$ be a hypergraph that is not $k$-Helly. By definition, there exists a subset $\W \subseteq \E$ such that $\cap_{e\in \W} e =\emptyset$ and $e_1\cap \ldots \cap e_{k}\neq \emptyset$ for any $e_1,\ldots, e_{k}\in \W$. Any such family is called a  a $k$-{\em witness} set of $G$. For every $V'\subset V$, define the {\em restriction} of $G$ to $V'$ as $G|_{V'} = (V',\E|_{V'})$, where $\E|_{V'} = \{e\cap V' | e\in \E\}$. With these definitions we can prove an upper bound on the Helly number of any hypergraph:

\begin{theorem}\label{theo_upper}
Let $G=(V,\E)$ be a hypergraph. If $|e|\leq k$ $\forall e \in \E$, then $G$ is $(k+1)$-Helly.
\end{theorem}
\begin{proof}
We will show the result by induction on $k$. Observe that the claim for $k=0$ is trivial, hence we focus on the induction step. Assume otherwise: let $\W \subseteq \E$ be a $(k+1)$-witness set, and $e$ be an edge of maximum size among those of $\W$ (by hypothesis we know that $|e|\leq k$). 

Consider the hypergraph $G' = (e,\W|_e \setminus \{e\})$ (that is, we disregard all other vertices except those contained in $e$). Since $|e|\leq k$, its intersection with any other edge of $\W$ must be of size at most $k-1$. Furthermore, every $k$-tuple of edges in $G'$ have a common intersection (since every $k+1$ tuple in $\W$ including $e$ had a common intersection). Therefore, by induction $G'$ is $k$-Helly. In particular all edges in $G'$ have a common intersection, which by construction intersects $e$ and contradicts the witness property.
\end{proof}

\begin{cor}\label{cor_upper}
Any polyomino $P$ satisfies $\H(P)\leq |P|+1$.
\end{cor}
The proof is direct from the fact that the associated hypergraph is $|P|$-uniform. We also note that the bound of Corollary \ref{cor_upper} is tight: the polyomino $\{(0,0), (1,0)$,$(0,1)\}$ (commonly referred as {\em El} \cite{b-cgttt-08}) has cardinality $3$ and contains the small empty quadrant structure. In particular, by Lemma \ref{lem:seqex} its Helly number must be at least $4$.



In the following we give a few more tools to use when proving that a given hypergraph is $k$-Helly (or equivalently, that there cannot exist a $k$-witness). 

\begin{lemma}\label{lem_notk}
Any $k$-witness $\W$ of a hypergraph $G$ satisfies $|\W|\geq k+1$ and $|e_1\cap \ldots \cap e_\ell|\geq k-\ell+1$ for all $e_1,\ldots, e_\ell \in \W$.
\end{lemma}
\begin{proof}
Observe that the first claim is trivial, since if $\W$ has size $k$ or less it cannot have an empty intersection. The proof of the second claim is by contradiction: assume otherwise and let $e_1,\ldots, e_\ell \in \W$ such that such that $e_1\cap \ldots \cap e_\ell=\{v_1,\ldots, v_m\}$ for some $m\leq k-\ell$. Since $\cap_{e\in \W} e =\emptyset$, for any $i\leq k-\ell$ there exists $f_i\in \W$ such that $v_i\not \in f_i$. 

Consider now the intersection of $e_1 \cap \ldots \cap e_\ell \cap f_1\cap \ldots \cap f_m$: by construction, this set is empty. Moreover, the size of the set $\{e_1,\ldots, e_\ell,f_1,\ldots, f_m\}$ is at most $\ell+ m\leq \ell +k-\ell=k$, which contradicts the witness property of $\W$.
\end{proof}



For any hypergraph $G$ and vertex $v\in V$, we define $c_v=\{e \in \W, v \in e\}$ as the edges that contain $v$. In the following we show that we can ignore vertices that are not heavily covered.

\begin{lemma}\label{lem_heavycover}
Let  $\W$ be a $k$-witness set of $G$ and let $V' = \{v \in V, |c_v| \geq k\}$. The set $\W|_{V'}$ is a k-witness for $G|_{V'}$.
\end{lemma}
\begin{proof}
Observe that $\cap_{e\in \W}e=\emptyset$ implies $\cap_{e\in \W|_{V'}}e=\emptyset$. Hence, it suffices to show that $e_1 \cap \ldots \cap e_k \cap V' \neq \emptyset$, for any $e_1,\ldots, e_k\in \W$,

Let $S=e_1\cap\ldots\cap e_k$. Observe that, since $\W$ is a witness set, we have $S\neq \emptyset$. Moreover all points of $S$ are covered by at least $k$ hyperedges (since they are contained in $e_1, \ldots, e_k$). Hence we have  $S\subseteq V'$. In particular, we obtain $e_1 \cap \ldots \cap e_k = e_1 \cap \ldots \cap e_k \cap V' \neq \emptyset$ which proves the Lemma.
\end{proof}

Lemma \ref{lem_notk} gives a lower bound on the size of a witness set. We use a similar reasoning to find an upper bound as well:

\begin{lemma}\label{lem_boundsize}
Let $G$ be any hypergraph such that $\H(G)=k$. There exists a $(k-1)$-witness set $\W\subseteq \E$ of $P$ such that $|\W|=k$.
\end{lemma}
\begin{proof}
Let $\W_{\min}$ be the $(k-1)$-witness set of smallest size (pick any arbitrarily if many exist) and let $m=|\W_{\min}|$. By Lemma \ref{lem_notk} we have $m\geq k$. If $m=k$ we are done, thus we focus in the $m>k$ case. 

By minimality of $\W_{\min}$, there cannot exist a proper subset $\W' \subset \W_{\min}$ such that $\cap_{A\in \W'} A =\emptyset$ (otherwise we would have a witness set of smaller size). In particular, any subset $\{e_1,\ldots, e_k\} \subset \W_{\min}$ must have non-empty intersection. Since $G$ is $k$-Helly, we have $\cap_{e\in \W_{\min}} e \neq\emptyset$ which contradicts the witness property.
\end{proof}


\section{Higher Helly Numbers}\label{sechihg}
In the following we use the above tools to show the existence of polyominoes of Helly number $k$ (for any $k\geq 5$). For any $q\in \N$, let $F_q$ be the union of rectangles $[-\lfloor q/2 \rfloor,-1]\times [0,0]$, $[1,q]\times [0,0]$ and $[-1,1]\times [1,1]$. Observe that $|F_q|=\lfloor3q/2\rfloor+3$, see Figure \ref{fig_khelly}.

\begin{figure}[tb]
\center
\includegraphics[width=0.6\textwidth]{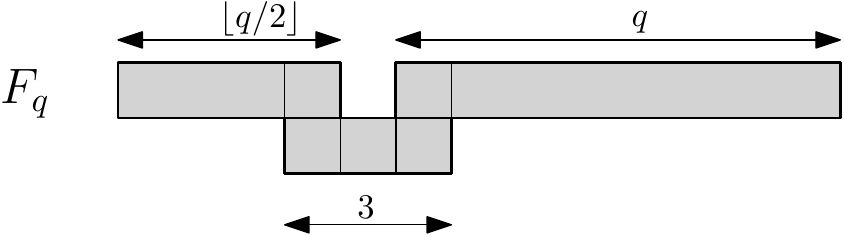}
\caption{Polyomino $F_{q}$. In Section \ref{sechihg} we show that $F_q$ has Helly number $q+1$ for any $q\geq 4$.}
\label{fig_khelly}
\end{figure}

\begin{figure}[tb]
\center
\includegraphics[width=.6\textwidth]{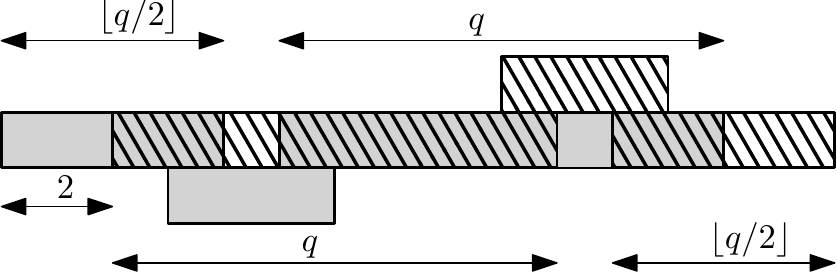}
\caption{Polyominoes $A_0$ (solid) and $B_2$ (dashed). In the example $q=8$.}
\label{fig_witness}
\end{figure}

\begin{figure}[tb]
\center
\includegraphics[width=0.6\textwidth]{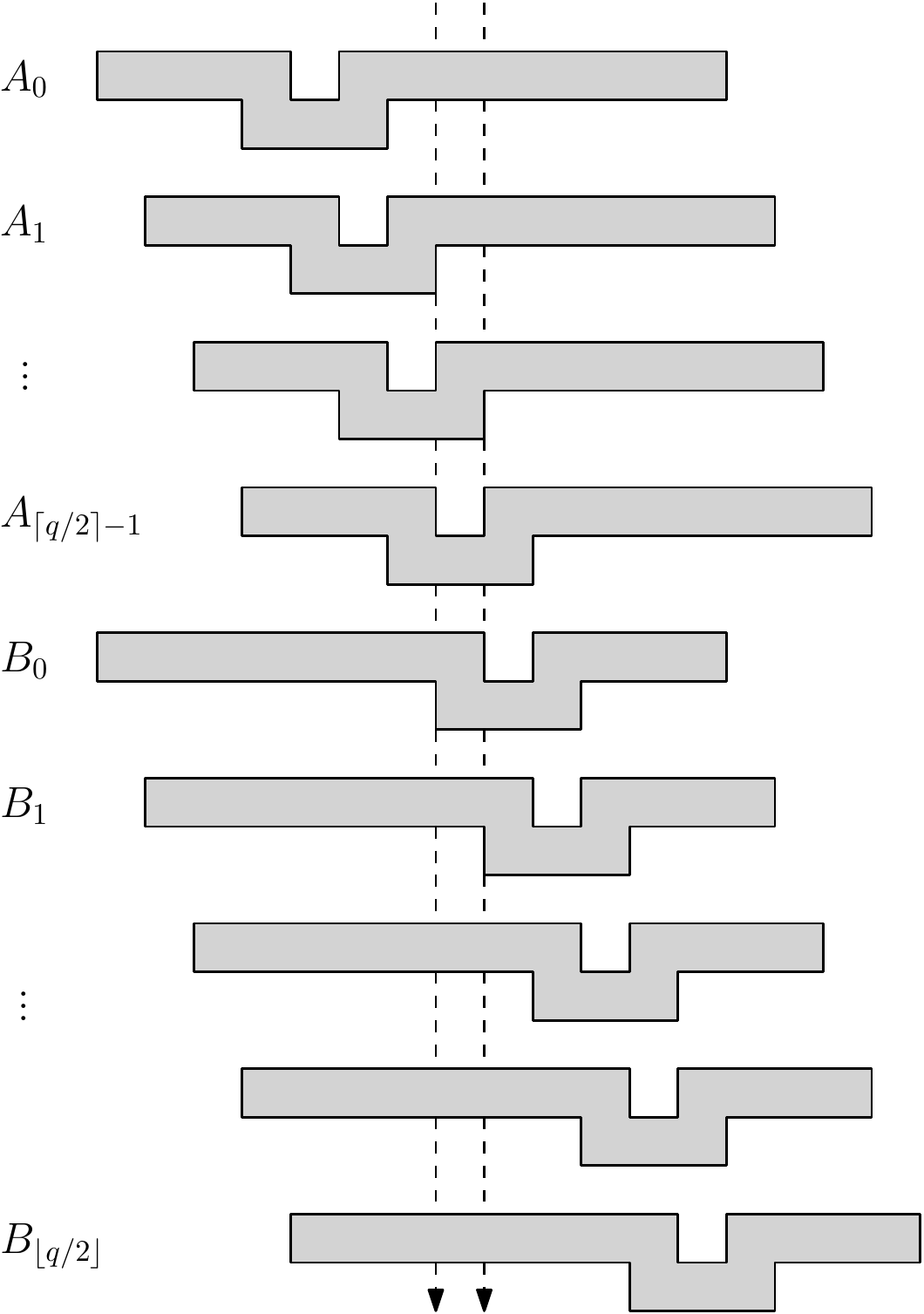}
\caption{$q$-Witness set for polyomino $F_q$ (for clarity, each of the copies has been shifted vertically). Observe that, although the intersection of the witness set is empty, any $q$ elements of the set have nonempty intersection. In the figure, we depicted with a vertical strip the point that is contained in all polyominoes except $A_{\lceil q/2 \rceil}-1$.}
\label{fig_counterw}
\end{figure}

\begin{lemma}\label{lem_gadget}
For any $q\geq 4$, we have $\H(F_q)=q+1$. 
\end{lemma}
\begin{proof}
We show the lower bound by constructing a $q$-witness set $\W$ of $F_q$. For any $i\leq q$, we define $A_i$ as the copy of $F_q$ translated such that the leftmost point is at position $(i,0)$. Analogously, we define polyomino $B_i$ as the $180$-degree rotation of  $F_q$ translated so as the leftmost point is at position $(i,0)$ (see Figure \ref{fig_witness}). Consider now set $\W=\{A_0,\ldots, A_{\lceil q/2 \rceil-1}, B_{0}, B_0, \ldots, B_{\lfloor q/2 \rfloor}\}$; observe that $|\W|=\lceil q/2 \rceil+{\lfloor q/2 \rfloor}+1=q+1$ and that the intersection between polyominoes $A_i$ and $B_j$ is in the rectangle $[0,\lfloor 3q/2 \rfloor]\times [0,0]$ (for any $i$ and $j$). 

More interestingly, for any $0 \leq i \leq \lceil q/2 \rceil-1$, polyomino $A_i$ does not contain point $(\lfloor q/2 \rfloor+i,0)$ (and this point is contained in all other polyominoes). The same result holds for polyomino $B_i$: for any $0 \leq i \leq \lfloor q/2 \rfloor$, point $(q+i,0)$ is contained in all polyominoes except $B_i$. In particular, we have $\cap_{C\in \W} C =\emptyset$ and any subset of size $q$ has nonempty intersection (see Figure \ref{fig_counterw}). Hence, $\W$ is a $q$-witness set of $F_q$. 

In order to finish the proof of the Lemma, we must show that polyomino $F_q$ indeed is $(q+1)$-Helly. Assume that $F_q$ is not $(q+1)$-Helly. Let $\W$ be a $(q+1)$-witness set and let $A$ be the leftmost copy of $F_q$ in $\W$ (pick any arbitrarily if more than one exist). Without loss of generality, we can assume that $A=A_0$. By Lemma \ref{lem_notk}, there must exist at least $q+1$ other copies $A$ of $F_q$ such that $|A\cap A_0|\geq q$. 

\begin{figure}
\center
\includegraphics[width=\textwidth]{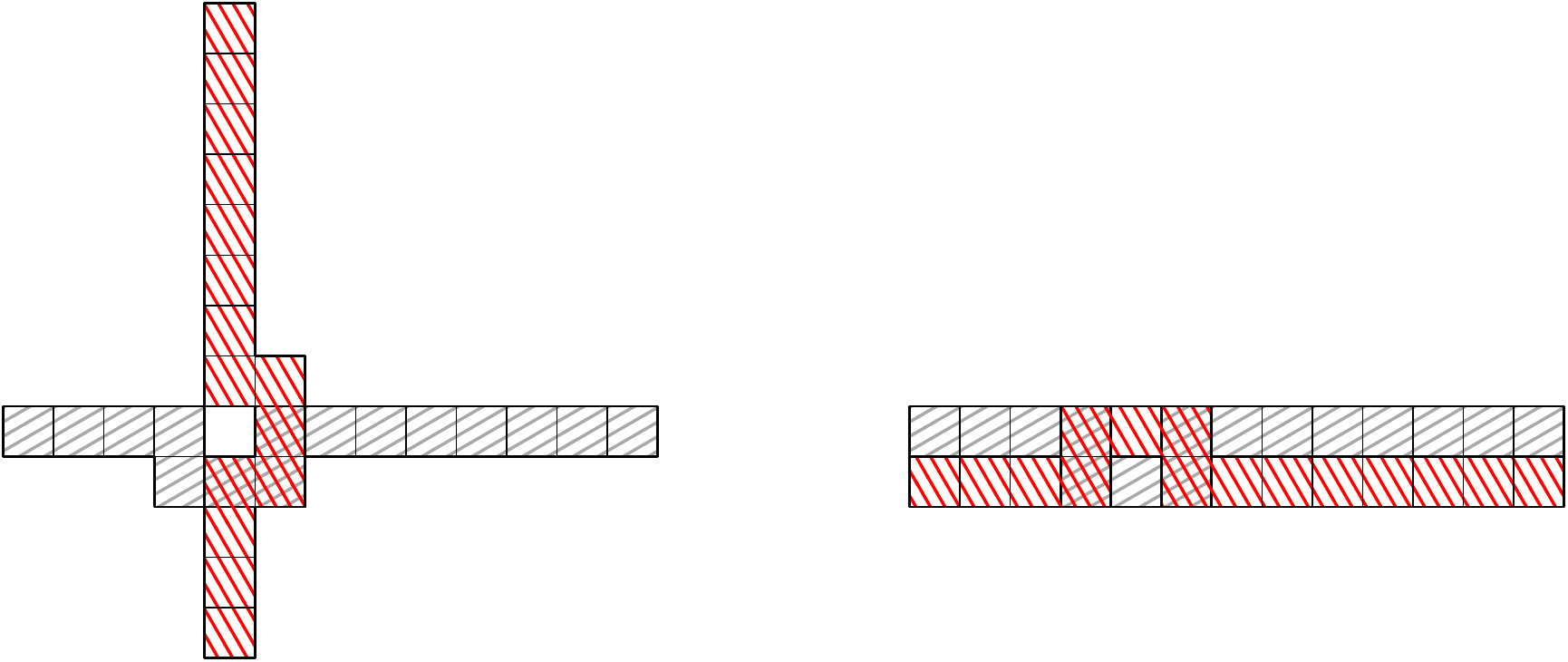}
\caption{Regardless of the value of $q$, 90 or 270-degrees rotation of $F_q$ can share at most three points (left). Likewise, two copies that are flipped across the horizontal axis can only share four points (right).}
\label{fig_common}
\end{figure}

First notice that if any two copies of the polyomino do not align their longest segment horizontally, they only have an intersection of size at most $4$ with $A_0$. Moreover, the only case when this intersection has size $4$ is if they are two copies flipped across the horizontal axis. In the latter case, any further copy can have an intersection of size at most $3$ with at least one of those two copies (see Figure \ref{fig_common}). Since in either case we obtain a contradiction with Lemma \ref{lem_notk} and the fact that $q\geq 4$, we can assume that for any $q+1$-witness set, all copies of $\W$ are aligned horizontally. 

Consider now the 3 lower points $(\lfloor q/2 \rfloor-1,-1),(\lfloor q/2 \rfloor,-1)$ and $(\lfloor q/2 \rfloor+1,-1)$ of $A_0$. Since $A_0$ is the leftmost copy of $P$ and $q\geq 4$ and copies are aligned horizontally, the three points can only be covered by at most two other copies ($A_1$ and $A_2$). Therefore we apply Lemma \ref{lem_heavycover} to show that any $(q+1)$-witness set of $\Z^2$ would be a witness set of $\Z^2 \setminus \{(\lfloor q/2 \rfloor-1,-1),(\lfloor q/2 \rfloor,-1),(\lfloor q/2 \rfloor+1,-1)\}$. Thus, we focus our attention in the rectangle $[0,\lfloor3q/2\rfloor]\times [0,0]$. 

Observe that, since we are considering only this rectangle, the extra copies caused by reflections across the horizontal axis are eliminated because they become the same hyperedge in the restricted hypergraph. Hence, all elements of $\W$ must be of the form $A_i$ or $B_j$ for some $i,j\geq 0$. Also notice that we have $|A_0 \cap A_i|\geq q$ if and only if $i\in \{1,\ldots, \lfloor q/2 \rfloor -1\}$ (provided that $q\geq 4$). Analogously, if $q\geq 2$ we have $| A_0 \cap B_j|\geq q \Leftrightarrow j\in \{0,\ldots, \lfloor q/2 \rfloor -1\}$. In particular, the set $\W$ can have at most $2\lfloor q/2 \rfloor $ elements, hence there cannot exist a $(q+1)$-witness set. 



\end{proof}






\begin{theorem}
For any $k\in\N$ such that $k\neq 1,3$, there exists a polyomino $P$ such that $\H (P)=k$.
\end{theorem}
\section{Experimental Results}\label{sec_imple}

As a complement to our research, we computed the Helly number of all polyominoes of small size with the help of a computer. The algorithm uses the results of Section \ref{sec_upper} and runs in exponential time, testing all possible witness sets. 

Specifically, we construct all polyominoes of size at most $15$ using the method of Redermeier \cite{r-cpyaa-81}. For each generated polyomino $P$, we tested whether or not it satisfies the Helly property $k$. In order to do so, we compute all copies of $P$ that have at least $k-1$ cells in common with a fixed polyomino, and store all such copies in a set $\mathcal{C}$. Once this set has been computed, we consider of its subsets and consider them as candidate witness sets.  Whenever a counterexample witness is found, we  can certify that the Helly number of the given polyomino is higher. Otherwise, we obtain an upper bound of its Helly number. Combining this approach with a binary search on $k$ gives us a method to compute the exact Helly number of all polyominoes of small size. 

Results of the execution can be seen in Table \ref{tab:result}, where $\mathcal{P}_n$ denotes the set containing all polyominoes of size $n$. We note that, as opposed to what one might expect, there is no monotonicity of any column, nor unimodality of any row. In Figures \ref{fig_examp} and \ref{fig13} we show some interesting polyominoes. More details of the implementation and results of the execution can be found in \cite{takumijap}.

\begin{figure}
\center
\includegraphics[width=\textwidth]{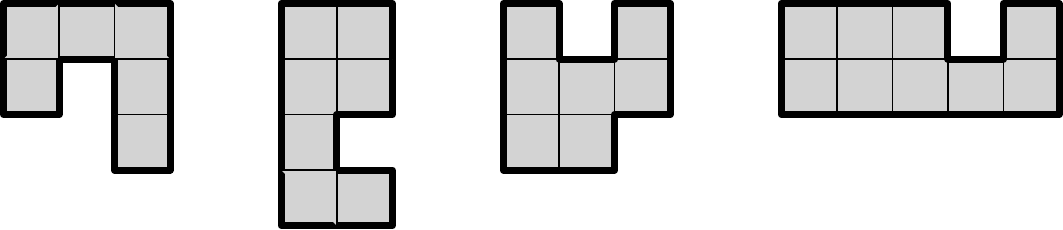}
\caption{Several polyominoes $P$ of Helly number $|P|-1$.}
\label{fig_examp}
\end{figure}

\begin{figure}
\center
\includegraphics[width=\textwidth]{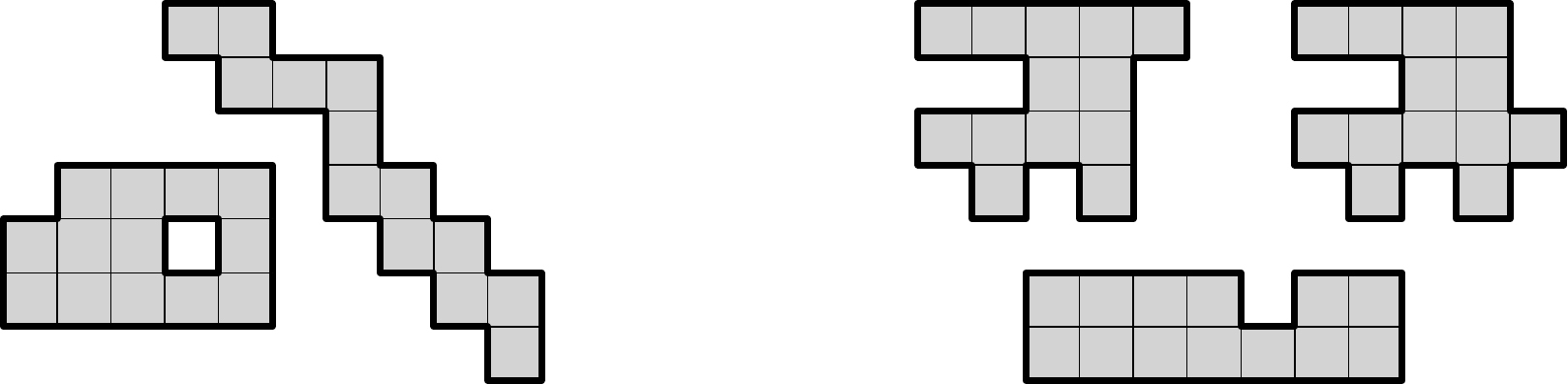}
\caption{Polyominoes of size $13$ of Helly number $9$ and $10$ (left and right, respectively).}
\label{fig13}
\end{figure}

\begin{landscape}

\begin{table}[t]
	\centering
	\label{tab:result}
	\begin{tabular}{|r|r|rrrrrrrr|r|}
		\hline
		&  & \multicolumn{8}{|c|}{Helly number} & computation\\
		\multicolumn{1}{|c|}{$n$} &
		\multicolumn{1}{|c|}{$|\mathcal{P}_n|$} &
		\makebox[4em][r]{2} & 
		\makebox[4em][r]{4} & 
		\makebox[4em][r]{5} & 
		\makebox[4em][r]{6} & 
		\makebox[4em][r]{7} & 
		\makebox[4em][r]{8} & 
		\makebox[4em][r]{9} & 
		\makebox[4em][r]{10} & 
		\multicolumn{1}{|c|}{time (sec.)} \\
		\hline
		1 & 1 & 1 & 0 & 0 & 0 & 0 & 0 & 0 & 0 & 0\\
		2 & 1 & 1 & 0 & 0 & 0 & 0 & 0 & 0 & 0 & 0\\
		3 & 2 & 1 & 1 & 0 & 0 & 0 & 0 & 0 & 0 & 0\\
		4 & 5 & 2 & 3 & 0 & 0 & 0 & 0 & 0 & 0 & 0\\
		5 & 12 & 1 & 11 & 0 & 0 & 0 & 0 & 0 & 0 & 0\\
		6 & 35 & 2 & 32 & 1 & 0 & 0 & 0 & 0 & 0 & 0.02 \\
		7 & 108 & 1 & 90 & 15 & 2 & 0 & 0 & 0 & 0 & 0.29 \\
		8 & 369 & 2 & 226 & 107 & 34 & 0 & 0 & 0 & 0 & 2.97 \\
		9 & 1285 & 2 & 526 & 479 & 277 & 0 & 1 & 0 & 0 & 20.62 \\
		10 & 4655 & 2 & 1173 & 1742 & 1684 & 8 & 46 & 0 & 0 & 128.95 \\
		11 & 17073 & 1 & 2378 & 5438 & 8639 & 124 & 493 & 0 & 0 & 674.67 \\
		12 & 63600 & 3 & 4627 & 15246 & 38883 & 1287 & 3554 & 0 & 0 & 3384.50 \\
		13 & 238591 & 1 & 8465 & 39286 & 159516 & 10110 & 21208 & 2 & 3 & 16059.70 \\
		14 & 901971 & 2 & 14941 & 94638 & 612251 & 65680 & 114338 & 74 & 47 & 77328.24 \\
		\hline
	\end{tabular}
		\caption{Distribution of the polyominoes of Helly number $k$ as a function of its size $n$ (for $n\leq 15$). Recall that $\mathcal{P}_n$ denotes the set containing all polyominoes of size $n$}	
\end{table}
\end{landscape}

\section{Conclusion}  
In this paper we have completely characterized for which values of $k$ there exist polyominoes of Helly number $k$. It is easy to see that the algorithm used in Section \ref{sec_imple} is exponential in the size of the polyomino. Hence, an interesting open problem is finding an efficient method that computes the Helly number of a given polyomino. Although it is known that designing a general algorithm that works for hypergraphs is difficult \cite{hypersurvey}, we wonder whether or not one can devise an algorithm that runs in polynomial time for polyominoes. 

The concept of {\em strong} $k$-Helly was introduced in \cite{gj-eigpit-85}. In this generalization, polyomino $P$ satisfies the {\em strong} $k$-Helly property, if for any finite family $\mathcal{A}$ of copies of $P$, there exist $A_1,\ldots, A_k\in\mathcal{A}$ such that $\cap_{A\in\mathcal{A}}A=A_1\cap \ldots \cap A_k$. Also, we note that we defined a copy of $P$ as any image of $P$ with respect to translations and the 8 symmetries of the square. It is easy to see that our results do not hold if we only consider the strong Helly property (instead of the classic definition), or if we only allow translations (or translations and rotations). For example, rectangles have Helly number $2$, but it is easy to see that they have strong Helly number $4$. It would be interesting to see how much can the Helly number of a given polyomino change when considering the alternative definitions. In particular, does there exist a polyomino $P$ whose Helly number dramatically increases if we only forbid the rotation and/or symmetry operations? or if we consider the strong Helly definition instead?

\section*{Acknowledgements}
The authors would like to thank Takumi Morishita for implementing and executing the algorithm of Section \ref{sec_imple}.

\small 

\end{document}